\documentclass[pra,superscriptaddress,nofootinbib,twocolumn,showpacs]{revtex4-1}
\usepackage{amsmath,graphicx,epsfig,color,amsfonts,amssymb,bbm}
\usepackage{amsthm}

\renewcommand{\P}{\mathcal{P}}

\newcommand{\ket}[1]{|#1\rangle}

\newcommand{\vecx}{\vec{x}}
\newcommand{\vecr}{\vec{r}}
\newcommand{\vecn}{{\vec{n}}}

\newcommand{\id}{\mathbbm{1}}
\renewcommand{\S}{\mathcal{S}}
\newcommand{\Q}{\mathcal{Q}}
\newcommand{\Qn}{\Q_\vecn}
\newcommand{\Sn}{\S_n}
\newcommand{\Sq}{\Sn^{\mbox{\tiny $\Q$}}}

\newcommand{\SqM}{\Sn^{\mbox{\tiny $\Q$},*}}
\newcommand{\SqkPrM}{\S^{\mbox{\tiny $\Q$},*}_{\mbox{\tiny $k$-pr.}}}
\newcommand{\SqkM}[1]{\S_{#1}^{\mbox{\tiny $\Q$},*}}

\newcommand{\I}{\mathcal{I}}
\newcommand{\Mnk}{\mathcal{M}^k_n}

\newcommand{\ubG}[1]{u_#1(\zeta_#1)}
\newcommand{\ub}[2]{u_#1(#2)}
\newcommand{\ubGg}[1]{\ub{{\vecn}}{#1}}
\newcommand{\ubg}{\ubGg{\zeta_{\vecn}}}
\newcommand{\nmax}{{n^*}}

\newcommand{\CorA}[1]{\mu_{#1}}
\newcommand{\CorAG}{\CorA{\vec{n}}}
\newcommand{\CorL}[1]{\zeta_{#1}}
\newcommand{\CorLG}{\CorL{\vec{n}}}

\newcommand{\proj}[1]{\left| #1 \right\rangle\!\!\left\langle #1 \right|}
\DeclareMathOperator{\tr}{tr}

\definecolor{nblue}{rgb}{0.2,0.2,0.9}
\definecolor{ngreen}{rgb}{0.2,0.6,0.2}
\definecolor{nred}{rgb}{0.8,0.2,0.2}
\definecolor{nblack}{rgb}{0,0,0}

\usepackage{hyperref}

\newtheorem{theorem}{Theorem}
\newtheorem{lemma}{Lemma}

\renewcommand{\L}{\mathcal{L}}
\begin{document}

\title{Family of Bell-like inequalities as \\ device-independent witnesses for entanglement depth}

\author{Yeong-Cherng~Liang}
\email{ycliang@mail.ncku.edu.tw}
\affiliation{Department of Physics, National Cheng Kung University, Tainan 701, Taiwan.}
\affiliation{Institute for Theoretical Physics, ETH Zurich, 8093 Zurich, Switzerland.}
\author{Denis Rosset}
\affiliation{Group of Applied Physics, University of Geneva, CH-1211 Geneva 4, Switzerland.}
\author{Jean-Daniel Bancal}
\affiliation{Centre for Quantum Technologies, National University of Singapore, 3 Science Drive 2, Singapore 117543.}
\author{Gilles P\"utz}
\affiliation{Group of Applied Physics, University of Geneva, CH-1211 Geneva 4, Switzerland.}
\author{Tomer Jack Barnea}
\affiliation{Group of Applied Physics, University of Geneva, CH-1211 Geneva 4, Switzerland.}
\author{Nicolas Gisin}
\affiliation{Group of Applied Physics, University of Geneva, CH-1211 Geneva 4, Switzerland.}

\date{\today}
\pacs{03.65.Ud, 03.67.Mn}

\begin{abstract}
We present a simple family of Bell inequalities applicable to a scenario involving arbitrarily many parties, each of which performs two binary-outcome measurements. We show that these inequalities are members of the complete set of full-correlation Bell inequalities discovered by Werner-Wolf-\.{Z}ukowski-Brukner. For scenarios involving a small number of parties, we further verify that these inequalities are facet-defining for the convex set of Bell-local correlations. Moreover, we show that the amount of quantum violation of these inequalities naturally manifests the extent to which the underlying system is genuinely many-body entangled. In other words, our Bell inequalities, when supplemented with the appropriate quantum bounds, naturally serve as device-independent witnesses for entanglement depth, allowing one to certify genuine $k$-partite entanglement in an arbitrary $n\ge k$-partite scenario without relying on any assumption about the measurements being performed, nor the dimension of the underlying physical system. A brief comparison is made between our witnesses and those based on some other Bell inequalities, as well as the quantum Fisher information. A family of witnesses for genuine $k$-partite nonlocality applicable to an arbitrary $n\ge k$-partite scenario based on our Bell inequalities is also presented.
\end{abstract}

\maketitle

One of the most important no-go theorems in physics concerns the impossibility to reproduce all quantum mechanical predictions using {\em any} locally-causal theory~\cite{Bell:Book} --- a fact commonly referred to as Bell's theorem~\cite{J.S.Bell:1964}. An important observation leading to this celebrated result is that measurement statistics allowed by such theories must satisfy  constraints in the form of an inequality, a {\em Bell inequality}. Since these inequalities only involve experimentally accessible quantities, their violation --- a manifestation of Bell-nonlocality~\cite{Brunner:RMP} --- can be, and has been (modulo some arguably implausible loopholes~\cite{Larsson:Loopholes}) empirically demonstrated (see, e.g.,~\cite{Brunner:RMP,Larsson:Loopholes,BellExp} and references therein). 

Clearly, Bell inequalities played an instrumental role in the aforementioned discovery.
Remarkably, they also find applications in numerous quantum information and communication tasks, e.g.,  in quantum key distribution involving untrusted devices~\cite{A.K.Ekert:PRL:1991,DIQKD, DIQKD2}, in the reduction of communication complexity~\cite{CommComplexity}, in the expansion of trusted random numbers~\cite{rand_colbeck, BIV:Randomness}, in certifying the Hilbert space dimension of physical systems~\cite{DimWitness,Navascues:PRX}, in self-testing~\cite{mayers,McKague:JPA:2012,Yang:PRA:2013,Swap,SelfTesting} of quantum devices, in witnessing~\cite{DIEW,K.F.Pal:1102.4320,GUBI} and quantifying~\cite{Verstraete:2002,SDIBE,Moroder:PRL:PPT,Toth:1409.3806} (multipartite) quantum entanglement using untrusted devices etc. For a recent review on these and other applications, see~\cite{Brunner:RMP}.

Identifying interesting or useful Bell inequalities is nonetheless by no means obvious. For instance, 
the approach of solving for the complete set of optimal, i.e., {\em facet-defining} Bell inequalities for a given experimental scenario --- though potentially useful for the identification of non-Bell-local (hereafter nonlocal) correlations --- typically produces a large number of inequalities with no apparent structure (see however~\cite{Rosset:JPA:2014} for some  progress made on classifying Bell inequalities). In contrast, carefully constructed Bell inequalities, such as the {\em families} of two-party Bell inequalities considered in~\cite{Pearle:PRD:1970, Braunstein:AP:1990,CGLMP,BKP}, have enabled us to conclude that certain correlations derived from maximally entangled states do not admit any {\em local content}~\cite{EPR2}, and that the prediction of quantum theory cannot be refined even when supplemented with hidden variables satisfying certain auxiliary assumptions~\cite{Colbeck:NC:2011}.

Going beyond the bipartite scenario, the family of Mermin-Ardehali-Belinskii-Klyshko (MABK) inequalities~\cite{MABK,Gisin:PLA:1998} is a prominent example of interesting family of Bell inequalities,
giving clear evidence that a macroscopic number of physical systems can still give rise to strongly nonclassical behavior (see also~\cite{RFFBIV,Wallman:PRA:2011}). Moreover, a sufficiently strong violation of the $n$-partite MABK inequalities can also be used to 
certify the presence of genuine $n$-partite entanglement in a {\em device-independent} manner, i.e., without relying on any assumption about the measurement device nor the dimension of the Hilbert space of the test systems (see, e.g.,~\cite{DIEW,K.F.Pal:1102.4320,GUBI,Nagata:PRL:2002,ANL}).

What about the possibility of identifying genuine $k$-partite entanglement in an $n$-partite scenario with $n>k>2$? This is the question of {\em entanglement depth}~\cite{Sorensen:PRL:2001}, or equivalently non-$k$-producibility~\cite{Guhne:NJP:2005} (see also~\cite{Seevinck:PRA:2001}), which both seek to identify the extent to which many-body entanglement is present in a multipartite quantum system. It is worth noting that genuine many-body entanglement is known to be essential, e.g.,  in achieving extreme spin squeezing~\cite{Sorensen:PRL:2001}, and also high sensitivity in some general metrology tasks~\cite{Fisher:ED}. For well-calibrated or trusted~\cite{Rosset:PRA:2012} measurement devices, there exist few criteria~\cite{Sorensen:PRL:2001,Seevinck:PRA:2001,Fisher:ED,Guhne:NJP:2005,Seevinck:PRA:2008,Durkin:PRL:2005,Levi:PRL:2013,Huber:PRA:2013,Lucke:PRL:2014}  to certify such  many-body entanglement. For example, an entanglement depth larger than 28 was recently demonstrated~\cite{Lucke:PRL:2014} using such a witness. However, the possibility of  certifying --- in a {\em device-independent} manner ---  genuine $k$-partite entanglement in an arbitrary $n\ge k>2$-partite scenario has remained elusive so far.

Here, we show that such robust certification in a scenario involving arbitrarily many parties is indeed possible --- using a novel family of $n$-partite Bell inequalities, and the characterized quantum violation of these inequalities by quantum states assuming only $k$-partite entanglement. Moreover, we show that, together with the appropriate bounds, these inequalities can also be used to witness genuine $k$-partite  nonlocality~\cite{GMNL} in an arbitrary $n$-partite scenario (with $n\ge k$). Since genuine $k$-partite entanglement is a prerequisite for the presence of genuine $k$-partite quantum nonlocality~\cite{Werner:PRA:1989,Curchod:IP}, witnesses for such multipartite nonlocality 
are also witnesses for entanglement depth. Let us stress, however, that our family of device-independent witnesses for entanglement depth  {\em do not} rely on the detection of such genuine multipartite nonlocality.

{\em A novel family of $n$-partite Bell inequalities}.- Consider a Bell-type experiment involving $n$ spatially-separated parties (labeled by $i\in\{1,2,\ldots,n\}$),  each of them performing two  binary-outcome measurements. We denote the measurement setting  of the $i$-th party by $x_i\in\{0,1\}$, and the corresponding measurement outcome  by $a_i=\pm1$. The correlation between these measurement outcomes can  be summarized succinctly using the collection of joint conditional probability distributions $\{P(\vec{a}|\vec{x})\}$ where $\vec{a}=(a_1,a_2,\ldots,a_n)$ and $\vec{x}=(x_1,x_2,\ldots,x_n)$ are, respectively, $n$-component vectors describing the combination of measurement outcomes and measurement settings. In terms of the $n$-partite full correlators
	$E_n(\vec{x})=\sum_{a_1,a_2,\ldots,a_n} \prod_{i=1}^n a_i P(\vec{a}|\vec{x})$,
our family of $n$-partite Bell inequalities $\I_n$ reads as:
\begin{align}\label{Ineq:Sliwa7n}
	\I_n: \Sn&=2^{1-n}\left[\sum_{\vec{x}\in\{0,1\}^n} E_n(\vec{x})\right]-E_n(\vec{1}_n)\,  \stackrel{\mbox{\tiny $\L$}}{\le} 1
\end{align}
where $\vec{1}_n=(1,\ldots,1)$ is an $n$-bit string of ones and $\L$ signifies that the inequality holds for a locally-causal theory. For $n=2$, inequality~\eqref{Ineq:Sliwa7n} is the Clauser-Horne-Shimony-Holt  Bell inequality~\cite{CHSH}; 
for $n=3$, it is equivalent to the 7-th tripartite inequality of~\cite{Sliwa}. For general $n$, we show in Appendix~\ref{App:Facet} that $\I_n$ defines a facet~\cite{Gruenbaum:Book} of the $n$-partite {\em full-correlation} polytope  characterized by Werner-Wolf-\.Zukowski-Brukner~\cite{WW,ZB}, thus being a member of the $2^{2^n}$ Bell inequalities discovered therein. For $n\le 8$, we can further verify numerically that  $\I_n$ corresponds to a facet of the polytope of locally-causal correlations ---  a property which we conjecture to hold true for general $n$.

From~\cite{WW}, it thus follows that the {\em maximal quantum violation} of $\I_n$ --- denoted by $\SqM$ --- is attainable if each party measures the $\pm1$-outcome observables~\cite{WW} $A_{x_i=0}=\cos\alpha\,\sigma_x + \sin\alpha\,\sigma_y$ and $A_{x_i=1}=\cos(\varphi_i+\alpha)\,\sigma_x + \sin(\varphi_i+\alpha)\,\sigma_y$ 
for some judiciously chosen $\alpha$, $\varphi_i\in[0, 2\pi]$  on the $n$-partite Greenberger-Horne-Zeilinger (GHZ) state~\cite{GHZ} $\ket{\text{GHZ}_n}= \frac{1}{\sqrt{2}}\left(\ket{0}^{\otimes n} + \ket{1}^{\otimes n} \right)$.
For  $n\le 8$, we certified using  a converging hierarchy of semidefinite programs~\cite{QMP:Hierarchy1,QMP:Hierarchy2} that $\SqM$ can be achieved by further setting $\varphi_1=\varphi_2=\cdots=\varphi_n=\phi_n$ and $\alpha=-\frac{n-1}{2n}\phi_n$ for some $\phi_n\in[0,\frac{\pi}{2}]$. Explicitly, this ansatz gives the  quantum value
\begin{subequations}\label{Eq:QuantumMaximum}
\begin{equation}\label{Eq:SqmMaxSimplified}
	\Sq(\phi_n)=2 \cos^{n+1}\frac{\phi_n}{2} - \cos\left(\frac{n+1}{2} \phi_n\right),
\end{equation}
where the explicit analytic values of $\phi_n$ (for $n\le 7$) leading to $\SqM$ can be found in Appendix~\ref{App:phic} (see  Table~\ref{Tbl:QuantumMax} for the corresponding value of $\SqM$).
For larger values of $n$, the above observation and further numerical evidences  lead us to conjecture that 
\begin{equation}\label{Eq:ConjectureOptimum}
	\SqM=\max_{\phi_n} \Sq(\phi_n).
\end{equation}
\end{subequations}
Indeed, for sufficiently large $n$, this maximum value over $\phi_n$ is well approximated by setting $\phi_n=\frac{2\pi}{n}$,  thus giving $\max_{\phi_n} \Sq(\phi_n)\stackrel{n\to\infty}{\to} 3$, i.e., the {\em algebraic maximum} of $\S_{\infty}$.\footnote{The algebraic maximum of $\Sn$ is the maximal value of  $\Sn$ attainable by all legitimate conditional probability distributions. As $n\to\infty$, the quantum violation~\eqref{Eq:QuantumMaximum} is thus as strong as that allowed by, for instance, signaling correlations.}

\begin{table}[h!]
\begin{ruledtabular}
\begin{tabular}{c|ccccccccc}
$n$  & $2$ & $3$
& $4$ & $5$ & $6$ & $7$ & $8$ & $\infty$
\\ \hline
 $\SqM$ &   $\sqrt{2}$  & $\frac{5}{3}$ & 1.8428  &  1.9746  &   2.0777  &   2.1610  &   2.2299 & 3 \\
 $v_{n,1}^{\mbox{\tiny Ent}}$ &  $\frac{1}{\sqrt{2}}$ &  $\frac{3}{5}$ &   0.5427 &  0.5064  & 0.4813 & 0.4627 & 0.4485 & $\frac{1}{3}$\vspace{0.1cm}\\ \hline
 $\Sn^{*}$ &   $2$  & $\frac{5}{2}$ & 2.7500  &  2.8750  &   2.9375  &   2.9688  &   2.9844 & 3 \\
\end{tabular}
\caption{\label{Tbl:QuantumMax}  
Summary of the maximal quantum violation and the  critical visibility  $v_{n,1}^{\mbox{\tiny Ent}}$, i.e.,  the infimum of $v_n$ in Eq.~\eqref{Eq:rhoV} before the mixture stops violating $\I_n$. Also included in the table is the algebraic maximum of $\I_n$, denoted by $\Sn^*$.
}
\end{ruledtabular}
\end{table}

{\em Entanglement depth and $k$-producibility}.- To see how $\I_n$, or more precisely its quantum violation can witness entanglement depth, let us now briefly recall the notion of $k$-producibility~\cite{Guhne:NJP:2005}: an $n$-partite pure  state $\ket{\psi}=\bigotimes_{j=1}^m \ket{\varphi_j}$ is said to be $k$-producible if all of its constituent  states $\ket{\varphi_j}$ are {\em at most} $k$-partite. Analogously, a mixed state $\rho$ is said to be $k$-producible if it can be written as a convex mixture of $k$-producible pure states --- the set of $k$-producible quantum states is thus convex. Evidently, the production of a $k$-producible state only requires (up to) $k$-partite entanglement. In the following, we say that a quantum state has an entanglement depth of $k$ if it is $k$-producible but not $(k-1)$-producible.

{\em A family of device-independent witnesses for entanglement depth}.- It is well-known that the observed Bell-inequality violation of a quantum state $\rho$ immediately implies that $\rho$ is entangled~\cite{Werner:PRA:1989}, and hence has an entanglement depth of 2 or higher. Moreover, from the convexity of the set of $k$-producible quantum states, we see that --- when there is no restriction on the Hilbert space dimension --- the set of correlations that is due to $k$-producible quantum states is also convex. In particular, since $k$-producibility implies $k'$-producibility for all $k'\ge k$, one expects that quantum states having a larger entanglement depth may also lead to a stronger violation of any given $n$-partite Bell inequality (e.g., $\I_n$): this is the central intuition behind what we call {\em device-independent witnesses for entanglement depth} (DIWED) --- a violation of which implies some lower bound on the entanglement depth of the underlying state. To this end, let us denote by $\SqkPrM$ the maximal quantum violation of $\I_n$ attainable by $n$-partite quantum states having an entanglement depth of $k$. In general, one may expect $\SqkPrM$ to depend on both $n$ and $k$, but the algebraic structure of $\Sn$, cf. Eq.~\eqref{Ineq:Sliwa7n}, allows us to show otherwise.

\begin{theorem}\label{Thm:Main}
The maximal possible quantum violation of $\I_n$ by  $k$-producible quantum states, $\SqkPrM$, is independent of $n$ and equals to $\SqkM{k}$, the maximal possible quantum violation of $\I_k$.
\end{theorem}

The full proof of the theorem can be found in Appendix~\ref{App:ProofUB}. Here, let us show that $\SqkPrM\ge\SqkM{k}$.
Consider $n$ parties sharing the quantum state $\ket{\mbox{GHZ}_k}\otimes\ket{0}^{\otimes n-k}$ with the first $k$ parties performing the {\em optimal} local measurements leading to $\SqkM{k}$ while the rest of the parties always measure  the trivial observable $\id$. It then follows from Eq.~\eqref{Ineq:Sliwa7n} and Born's rule that the quantum value of $\Sn$ becomes $\SqkM{k}$. Since this is only one particular choice of quantum strategy, we must have $\SqkPrM\ge\SqkM{k}$. For instance, it is conceivable that with non-trivial local measurements on $\ket{\mbox{GHZ}_2}^{\otimes 2}$, a stronger violation of $\I_4$ could be obtained.
Theorem~\ref{Thm:Main}, however, dictates that this intuition is false. Indeed, the proof of the theorem  (Appendix~\ref{App:ProofUB}) suggests that to achieve the strongest quantum violation of $\I_n$ by $k$-producible quantum states, we should employ the above strategy of generating  optimal nonlocal correlation for only $k$ of the parties while leaving the rest of the $n-k$ parties with trivial correlations. 

The above theorem, together with the respective values of $\SqkM{k}$  [cf. Eq.~\eqref{Eq:QuantumMaximum} and Table~\ref{Tbl:QuantumMax}], then provides us with a family of DIWED:
\begin{equation}\label{Eq:EDWitnesses}
	\I_n^k: 2^{1-n}\sum_{\vec{x}\in\{0,1\}^n} E_n(\vec{x})-E_n(\vec{1}_n)\,  \stackrel{\stackrel{\mbox{\tiny $k$-producible}}{\mbox{\tiny states}}}{\le} \SqkM{k}.
\end{equation}
Since the upper bound $\SqkM{k}$ holds for {\em all} $n$-partite, $k$-producible quantum states of arbitrary Hilbert space dimensions, and arbitrary binary-outcome measurements performed by each party, the witness is device-independent in the sense that any observed violation of $\I^k_n$ by $\rho$ implies that $\rho$ is at least genuinely $(k+1)$-partite entangled, i.e., has an entanglement depth of at least $k+1$, regardless of the details of the measurement devices and the Hilbert space dimensions. For instance, a measured quantum value of $\I_n$ that is greater  than $\sqrt{2}$ and $\tfrac{5}{3}$ (cf. Table~\ref{Tbl:QuantumMax}) immediately implies, respectively, the presence of genuine tripartite and quadri-partite entanglement, regardless of the total number of parties $n$. For the noisy GHZ state
\begin{equation}\label{Eq:rhoV}
	\rho(v_n) = v_n\proj{\mbox{GHZ}_n}+(1-v_n)\frac{\id_{2^n}}{2^n}, 
\end{equation}
where $\id_{2^n}$ is the identity operator acting on $\mathbb{C}^{2n}$, such quantum violations then translate to the critical visibility of $v_n> v_{n,k}^{\mbox{\tiny Ent}}=\SqkM{k}/\SqkM{n}$ required for the device-independent certification of genuine $(k+1)$-partite entanglement via $\I^k_n$.

Let us emphasize again that the certification of genuine $(k+1)$-partite entanglement via $\I_n^k$ does not rely on the detection of genuine $(k+1)$-partite  nonlocality~\cite{GMNL,Curchod:IP}. Indeed, as we show in Appendix~\ref{App:NS}, the witnesses for genuine multipartite nonlocality~\cite{GMNL} corresponding to $\I_n$ read as:
\begin{equation}\label{Eq:EDWitnesses:NS}
	\I_n^{k,\mbox{\tiny NL}}: 2^{1-n}\sum_{\vec{x}\in\{0,1\}^n} E_n(\vec{x})-E_n(\vec{1}_n)\,  \stackrel{\mathcal{NS}_{n,k}}{\le} 3-2^{2-k},
\end{equation}
where $\mathcal{NS}_{n,k}$ signifies that the inequality holds for arbitrary $n$-partite correlations that are $k$-producible~\cite{Curchod:IP} (when assuming only non-signaling~\cite{NS,Barrett:PRA:2005} resources within each group).
Interestingly, as with quantum entanglement, the right-hand-side of inequality~\eqref{Eq:EDWitnesses:NS}  is simply the algebraic maximum of $\I_k$, which is achievable by a general $k$-partite non-signaling correlation.\footnote{Since $\Sn$ only involves a linear combination of full correlators, inequality~\eqref{Eq:EDWitnesses} also holds true even if we consider, instead, $k$-producible Svetlichny~\cite{Svetlichny} (signaling) correlations; see~\cite{Curchod:IP}.} For $n\le 8$, the explicit values of these algebraic maxima ($\Sn^*=3-2^{2-n}$) are clearly higher than the corresponding quantum bounds (see Table~\ref{Tbl:QuantumMax}). Thus, witnessing genuine $k$-partite entanglement via $\I_n^k$ does not rely on the detection of genuine $k$-partite nonlocality.

{\em Comparison with some other witnesses for entanglement depth}.- Given the intimate connection~\cite{Guhne:NJP:2005} between  $k$-producibility and $m$-separability,\footnote{A pure state is  $m$-separable if it can be written as the tensor product of $m$ constituent pure states. The definition for mixed states proceeds analogously. Thus, an $m$-separable state is also  $k$-producible for some $k\ge\lceil \frac{n}{m}\rceil$.}one expects that DIWED can also be constructed from other multipartite Bell inequalities where their $m$-separability properties are well-studied. Indeed, a series of investigations~\cite{Gisin:PLA:1998,MABKSeparability}  on the MABK inequalities have culminated in the following characterization~\cite{Nagata:PRL:2002,Yu:PRL:2003}: the maximal possible quantum violation of the $n$-partite MABK inequality by  $n$-partite, $m$-separable states ($m<n$) consisting of $L$ unentangled subsystems is~\cite{Nagata:PRL:2002,Yu:PRL:2003} $2^{(n+L-2m+1)/2}$. It thus follows that (see Appendix~\ref{App:MABK}) if $n\le 5$, the MABK inequalities give the following DIWED:
\begin{equation}\label{Ineq:MABKWitness}
	 \Mnk:2^{\frac{1-n}{2}}\!\!\!\!\!\!\sum_{\vec{x} \in\{0,1\}^n}\!\!\!\!
	 \cos\left[\frac{\pi}{4}\left(1-n+2{\bf x}\right)\right]\!E\left(\vec{x}\right) \stackrel{\stackrel{\mbox{\tiny $k$-producible}}{\mbox{\tiny states}}}{\le} 2^{\frac{k-1}{2}},
\end{equation}
where ${\bf x}=\sum_i x_i$ and we have made use of the compact representation of the MABK inequality obtained in~\cite{Wallman:PRA:2011}. Unfortunately, for $n\ge 6$, except for $k=2$ and $k=n-1$, the inequality given in Eq.~\eqref{Ineq:MABKWitness} generally does not hold for $k$-producible states (see Appendix~\ref{App:MABK} for details).

To compare the strength of $\I^k_n$ and $\Mnk$ in identifying the entanglement depth of various quantum states, we numerically optimized the quantum violation of these witnesses for the GHZ state $\ket{\mbox{GHZ$_n$}}$, the $n$-partite W-state~\cite{Dur:PRA:062314} $\ket{\mbox{W$_n$}}=\frac{1}{\sqrt{n}}(\ket{100\ldots0}+\ket{010\ldots0}+\ldots+\ket{000\ldots1})$, as well as the $n$-partite 1-dimensional cluster states~\cite{Briegel:PRL:2001} $\ket{C^-_n}= \prod_{i=1}^{n-1}CZ_{i,i+1}\ket{+}^{\otimes n}$ and $\ket{C^o_n}= CZ_{n,1}^{1-\delta_{n,2}}\ket{C^-_n}$ where $CZ_{i,j}={\rm diag}(1,1,1,-1)$ is the controlled $Z$ gate acting on the the $i$-th and the $j$-th qubit, and $\delta_{i,j}$ is the Kronecker delta.  A comparison between the best quantum violations found (see Appendix~\ref{App:QuantumViolation}) and the respective bounds associated with the witnesses, cf. Eq.~\eqref{Eq:EDWitnesses} and Eq.~\eqref{Ineq:MABKWitness}, then allows us to lower-bound  the entanglement depth of these states via those DIWED\footnote{All these states are known have an entanglement depth of $n$ (see, e.g.~\cite{Hein:PRA:2004,Toth:PRL:2005}).} (see Table~\ref{Tbl:ED}).
 Interestingly, neither of these witnesses appears to be strictly stronger than the other as each of them provides a better lower bound on entanglement depth than the other witness for a certain state. Moreover, for $\ket{C^o_n}$, it can be shown that~\cite{Florian:PC} the lower bounds originating from these DIWED even outperform those obtained from the {\em non-device-independent witnesses} based on quantum Fisher information~\cite{Fisher:ED}.

		\begin{table}[h!]
		\begin{ruledtabular}
		\begin{tabular}{c|c|llllll}
		$\ket{\psi}$ & Witness(es) &   $\bf 2$ & $\bf 3$ & $\bf 4$ & $\bf 5$ & $\bf 6$ & $\bf 7$  \vspace{0cm} \\ \hline\hline
		 $\ket{\mbox{GHZ$_n$}}$ & $\I^k_n$, $\Mnk$ &    $2^*$  & $3^*$ & $4^*$  &  $5^*$  &   $6^*$ & $7^*$ \vspace{0cm}\\ \hline
		 $\ket{\mbox{W$_n$}}$ & $\I^k_n$ & $2^*$  & $2$ & 2  &  2  &   2  & 2  \vspace{0cm}\\ 
		 $\ket{\mbox{W$_n$}}$ & $\Mnk$  &   $2^*$  & $3^*$ & 3  &  3 &  3 & 3  \vspace{0cm}\\ \hline
		 $\ket{C^-_n}$ & $\I^k_n$, $\Mnk$ &   $2^*$  & $3^*$ & 2  &  2  &  2   & 2 \vspace{0cm}\\ \hline
		 $\ket{C^o_n}$ & $\I^k_n$ &   $2^*$  & $3^*$ & 2  &  2  & 2 & 2   \vspace{0cm}\\
		 $\ket{C^o_n}$ & $\Mnk$ &  $2^*$  & $3^*$ & 2  &  1  &  1 & 1   \vspace{0cm}\\
		\end{tabular}
		\caption{\label{Tbl:ED}Lower bounds on entanglement depth (ED) certifiable by the violation of DIWED $\I^k_n$ and $\Mnk$. In the rightmost block, the  boldfaced numbers in the top row gives $n$ (the number of parties) whereas all integers underneath are the respective lower bounds on ED for the quantum state given in the leftmost column, using the witness(es) indicated in the second column. A tight lower bound is marked with an asterisk ($^*$).}
		\end{ruledtabular}		
		\end{table}

{\em Discussion}.- Obviously, for any given $n$ and $k$, the set of correlations that can arise from $k$-producible quantum states cannot be fully characterized by $\I_n^k$ alone (nor together with  $M_n^k$).\footnote{In general, each of these sets can only be fully characterized by an infinite number of such linear witnesses~\cite{Curchod:IP}.} Thus, one may ask if there exists a better DIWED, e.g.,  that closes the gap between the actual entanglement depth of $\ket{C_n^o}$ and the lower bound provided in Table~\ref{Tbl:ED}. To answer this, or more generally, the question of whether some observed correlation $\{P(\vec{a}|\vec{x})\}$ could have come from a $k$-producible quantum state, the hierarchy of semidefinite programs proposed in~\cite{Moroder:PRL:PPT} turns out to be well-suited. For completeness, we include the explicit form of these semidefinite programs  in Appendix~\ref{App:SDP}. Using this technique, it was found in~\cite{Curchod:JPA:2014} that all the 23,306 quadripartite Bell-like inequalities obtained therein are also legitimate DIWED for an entanglement depth of 2, some even for an entanglement depth of 3. Moreover, our numerical optimizations show that some of these inequalities can also be used to device-independently certify the genuine 4-partite entanglement present in $\ket{\mbox{W$_4$}}$, $\ket{C^-_4}$ and $\ket{C^o_4}$. Is it then always possible to find an appropriate DIWED to certify the entanglement depth of {\em any} pure entangled quantum state? Given the strong connection between nonlocality and pure entangled states (see, e.g.,~\cite{N.Gisin:PLA:1992,S.Popescu:PLA:1992,Yu:PRL:2012}), we are optimistic that the answer to the above question is positive.

Let us now comment on some other possibilities for future work. Naturally, a question that stems from our results is the typicality of Bell inequalities that are naturally suited for witnessing entanglement depth, in the sense of Theorem~\ref{Thm:Main}. To this end, we show in Appendix~\ref{App:OtherBellIneq} that  the family of DIWED given in Eq.~\eqref{Eq:EDWitnesses} actually belongs to an even more general family of DIWED --- $\I_n^k(\gamma)$ --- such that $\I_n^k(2)$ gives Eq.~\eqref{Eq:EDWitnesses}. The usefulness of this more general family of DIWED, however, remains to be investigated. Note also that apart from $\gamma=2$,  none of the Bell inequalities corresponding to $\I_n^k(\gamma)$ define a facet of the local polytope for general $n$. In contrast, as we show here, the combination of full correlators given by $\Sn$, cf. Eq.~\eqref{Ineq:Sliwa7n}, are natural both in the characterization of the set of locally-causal correlations, as well as the set of correlations allowed by $k$-producible quantum states, for arbitrary $k>1$.

On the other hand, since our witnesses involve the expectation value of $2^n$  different combinations of measurement settings, measuring these expectation values using only local measurement presents a great experimental challenge already for moderate values of $n$.\footnote{Although this scaling is still favorable compared with doing a full-state tomography of an $n$-qubit state.} Hence, it is certainly worth looking for other (families of) Bell inequalities where the corresponding DIWED only involve few expectation values but which may still share features of $\I_n$ given in Theorem~\ref{Thm:Main}. The families of Bell-like inequalities presented in~\cite{GUBI,Mermin:CGLMP} are some possible starting points for such an investigation and the numerical techniques that we detailed in Appendix~\ref{App:SDP} will be useful for this purpose. Note also that for any given positive integer $k$, Theorem 2 of~\cite{Curchod:IP} allows us to extend any given witness for $n\ge k$ parties to one for {\em arbitrarily} many parties while preserving the number of expectation values that need to be measured experimentally. From an experimental perspective, it will also be highly desirable to identify DIWED that only involve few-body correlators (cf. Bell inequalities given in~\cite{Jura:Science}), a problem that we leave for future research.

%%% acknowledgments

This work is supported by the  Swiss NCCR ``Quantum Science and Technology", the CHIST-ERA DIQIP, the ERC grant 258932, the Singapore National Research Foundation (partly through the Academic Research Fund Tier 3 MOE2012-T3-1-009) and the Singapore Ministry of Education. We gratefully acknowledge Florian Fr\"owis for enlightening discussion, and for sharing his computation of a lower bound on entanglement depth using quantum Fisher information.

%\end{acknowledgement}

%%% Appendices 
\appendix

\section{Proof that inequality~\eqref{Ineq:Sliwa7n} defines a facet of the local full-correlation polytope}
\label{App:Facet}

Inequality~\eqref{Ineq:Sliwa7n} can also be rewritten in the form
\begin{equation}
	\Sn=\sum_{\vec{x}\in\{0,1\}^n} \beta(\vecx)E(\vec{x})\le 1
\end{equation}
where $\beta(\vecx)=2^{1-n}-1$ when $\vecx=\vec{1}$ is the $n$-bit string of ones and $\beta(\vecx)=2^{1-n}$ otherwise. In what follows, we show that this inequality is indeed one of the $2^{2^n}$ full-correlation Bell  inequalities derived by Werner and Wolf~\cite{WW}, and hence a facet of the local full-correlation polytope for arbitrary $n$.

To prove the above claim, it is sufficient to prove that the function
\begin{equation}\label{eq:fr}
	f(\vecr)=\sum_{\vecx\in\{0,1\}^n}\beta(\vecx)(-1)^{\vecr\cdot\vecx}
\end{equation}
is indeed a $\pm1$-valued function of the $n$-bit string $\vecr$, as shown in~\cite{WW}. Let us start by proving the following Lemma.
\begin{lemma}
For any given $n$-bit string $\vecr$, it holds that $\sum_{\vecx} (-1)^{\vecr\cdot\vecx}=2^n\,\delta_{\vecr,\vec{0}}$. In other words, the sum vanishes unless $\vecr$ is exactly $n$ bits of 0.
\end{lemma}
\begin{proof}
For an $n$-bit string that is not identically 0, let us suppose, without loss of generality, that the first bit of $\vecr$, i.e., $r_1$ is 1, we see that
\begin{align}
	\sum_{\vecx\in\{0,1\}^n} (-1)^{\vecr\cdot\vecx}=&\sum_{x_1=0,1}\sum_{x_2,\ldots,x_n=0,1} (-1)^{r_1x_1+r_2x_2+\ldots +r_nx_n},\nonumber\\
	=&\sum_{x_1=0,1}(-1)^{x_1}\!\!\!\sum_{x_2,\ldots,x_n=0,1} (-1)^{r_2x_2+\ldots +r_nx_n},\nonumber\\
	=&\,0.
\end{align}
In contrast, if $\vecr=\vec{0}=(0,0,\ldots,0)$, it is easy to see that the sum is $2^n$.
\end{proof}

Now, let us prove that $f(r)$ as defined above indeed only takes values $\pm1$. To this end, let us denote by $\pi_{\vecr}$ the parity of $\vecr$, i.e., $\pi_{\vecr}=\oplus_{i} r_i$, and note that
\begin{equation}
\begin{split}
	f(\vecr)&=\beta(\vec{1})(-1)^{\pi_{\vecr}}+\beta(\vec{0})\sum_{\vecx\in\{0,1\}^n}(-1)^{\vecr\cdot\vecx}-\beta(\vec{0})(-1)^{\pi_{\vecr}}\\
			&=\left[\beta(\vec{1})-\beta(\vec{0})\right](-1)^{\pi_{\vecr}}+\beta(\vec{0})\,2^n\delta_{\vecr,\vec{0}}\\		
			&=-(-1)^{\pi_{\vecr}}+2\,\delta_{\vecr,\vec{0}}.		
\end{split}
\end{equation}
Hence,  $f(\vec{0})=1$, and $f(\vec{r}\neq\vec{0})=-1^{\pi_{\vecr+1}}$, i.e., $f(\vec{r})=\pm1$. Thus $\I_n$ is indeed a member of the $2^{2^n}$ full-correlation Bell inequality found by Werner and Wolf~\cite{WW}.

\section{Analytic expressions of $\phi_n$ and $\SqM$ }
\label{App:phic}

In Table~\ref{Tbl:Analytic} below, we provide, for $n\le 5$, the analytic expression of the  {\em optimal} $\phi_n$ that leads to the maximal quantum violation $\SqM$  of inequality~\eqref{Ineq:Sliwa7n}.

\begin{table}[h!]
\begin{ruledtabular}
\begin{tabular}{r|cc}
$n$ &  $\phi_n$  & $\SqM$ 
\\ \hline
  2 &   $\frac{\pi}{2}\approx1.5708$ & $\sqrt{2}$ \\ 
  3 &   $2\cos^{-1}\sqrt{\frac{2}{3}}\approx1.2310$ & $\frac{5}{3}$ \\
  4 &   $2 \cos^{-1}\sqrt{\frac{1}{14} (6 +  \sqrt{22} )}\approx1.0155$ & $\frac{2}{7} \sqrt{\frac{2}{7} (94 + 11 \sqrt{22})}$ \\
  5 &   $2 \cos^{-1}\sqrt{\frac{1}{15} (8 +  \sqrt{19} )}\approx0.8660$ & $\frac{1}{225} \sqrt{113 + 76 \sqrt{19}}$ \\  
\end{tabular}
\caption{\label{Tbl:Analytic} Analytic expressions of the optimal value of $\phi_n$ and the corresponding maximal quantum violation of inequality~\eqref{Ineq:Sliwa7n} for $n\le 5$.}
\end{ruledtabular}
\end{table}
For some higher values of $n$, it is still possible to solve for the optimal analytic value of $\phi_n$ in Eq.~\eqref{Eq:QuantumMaximum} using some trigonometric identities.
However, the resulting analytic expressions for $\phi_n$ and $\SqM$ quickly become very cumbersome and are not particularly insightful. For instance, with the help of Mathematica, we obtain the following analytic expressions of the optimal $\phi_6$ and $\phi_7$:
\begin{equation*}
\begin{split}
	\phi_6& =2 \cos^{-1}\sqrt{\tfrac{4}{93} \left[10+11\cos\left(\tfrac{1}{3}\tan^{-1}\tfrac{93\sqrt{186735}}{14107}\right) \right]}\\
	&\approx0.7559
\end{split}
\end{equation*}
and
\begin{equation*}
\begin{split}
	\phi_7& =2 \cos^{-1}\left\{\tfrac{2}{3}\sqrt{ \tfrac{1}{7}\left[8+\sqrt{46}\cos\left(\tfrac{1}{3}\tan^{-1}\tfrac{63\sqrt{383}}{193}\right)\right]}\right\}\\
	&\approx0.6713.
\end{split}
\end{equation*}

\section{Proof of Theorem~\ref{Thm:Main}}
\label{App:ProofUB}

Here, we give a proof of Theorem~\ref{Thm:Main} in the main text. Firstly, in Appendix~\ref{App:Notations}, we reformulate the problem in terms of some new notations introduced therein. Then, in Appendix~\ref{App:2DSpace}, we give some preliminary characterization of a 2-dimensional projection of the set of quantum correlations. A technical lemma that allows us to relate the set of quantum correlations admissible by $k$-producible quantum states and general $k$-partite quantum states in this 2-dimensional projection is provided in Appendix~\ref{App:Lemma}. The proof is then completed with some further characterizations of the set of quantum correlations in this 2-dimensional projection given in Appendix~\ref{App:FurtherCharacterization}.

\subsection{Notations and reformulation of the problem}
\label{App:Notations}
Let us define:
\begin{equation}\label{Eq:Axes}
	\mu_\ell\equiv \frac{1}{2^\ell}\sum_{\vec{x}\in \{ 0,1 \}^{\ell}} E_\ell(\vec{x}),\quad \zeta_\ell\equiv E_\ell(\vec{1}_\ell).
\end{equation}
For quantum correlations, these quantities can be expressed as
\begin{gather}
	\mu_{\ell}=\frac{1}{2^\ell} \sum_{\vec{x}\in \{ 0,1 \}^{\ell}} \tr\,\left(\rho\, A_{x_1}\otimes A_{x_2}\otimes \cdots \otimes A_{x_\ell}\right),\nonumber\\
	\zeta_{\ell}= \tr\,\left(\rho\, A_{x_1=1}\otimes A_{x_2=1}\otimes \cdots \otimes A_{x_\ell=1}\right),
	\label{Eq:QuantumCor}
\end{gather}
where $\rho$ is a quantum state and $\{\{A_{x_j}\}_{x_j=0,1}\}_{j=1}^\ell$ are dichotomic observables satisfying $A_{x_j}^2=\id$.

Since the Bell inequality $\I_n$ is a linear function of the expectation value $E_n(\vecx)$, and each $E_n(\vecx)$ is a linear function of $\rho$, we may without loss of generality consider pure quantum state $\rho=\proj{\psi}$ in determining the maximal value of the inequality attainable by ($k$-producible) quantum states. For a pure $n$-partite state $\ket{\psi_\vecn}=\bigotimes_{j=1}^m \ket{\varphi_j}$ where each constituent pure state  is itself $n_j$-partite, let us define $\vecn=(n_1,\ldots,n_m)$. We thus have
\begin{equation}\label{Eq:Axes2}
	\mu_{\vecn}=\prod_{i=1}^m \mu_{n_i}\quad \zeta_{\vecn} = \prod_{i=1}^m \zeta_{n_i},
\end{equation}
where we have abused the notations and used $\mu_{\vecn}$, $\zeta_\vecn$ to denote the analog of Eq.~\eqref{Eq:Axes} for such factorizable states.\footnote{
When the state is not factorizable, it should be understood that $\vecn={n}$, and thus  $\mu_{\vecn}=\mu_n$, $\zeta_\vecn=\zeta_n$.}

Let us denote by $\SqkM{\vecn}$ the maximal possible quantum violation of $\I_n$ achievable by quantum states having the tensor-product structure specified by  $\vecn$. In the notations introduced above, we have
\begin{equation}\label{Eq:vecnMax}
	\SqkM{\vecn}=\max_{\zeta_{\vecn},\mu_{\vecn}}\quad 2\mu_{\vecn}-\zeta_{\vecn}.
\end{equation}
The maximal quantum violation of $\I_n$ achievable by $n$-partite, $k$-producible states is thus:
\begin{equation}\label{Eq:k-prodMax}
	\SqkPrM=\max_\vecn \SqkM{\vecn}
\end{equation}
under the assumption of
\begin{equation}\label{Eq:nmax}
	\nmax\equiv\max_i\, n_i=k,
\end{equation}
whereas the maximal quantum violation of $\I_k$ reads as:
\begin{equation}
	\SqkM{k}=\max_{\zeta_{k},\mu_{k}}\quad 2\mu_{k}-\zeta_{k}.
\end{equation}
In this terminology, a proof of Theorem~\ref{Thm:Main}  thus consists of showing that $\SqkPrM=\SqkM{k}$, or more explicitly,
\begin{equation}\label{Eq:Goal}
	\max_\vecn \max_{\zeta_{\vecn},\mu_{\vecn}}\quad 2\mu_{\vecn}-\zeta_{\vecn} = \max_{\zeta_\nmax,\mu_\nmax}\quad 2\mu_\nmax-\zeta_\nmax,
\end{equation}
under the assumption of Eq.~\eqref{Eq:nmax}.
Note that the maximizations  on both sides of Eq.~\eqref{Eq:Goal} are to be carried out over all legitimate pairs of quantum distributions $(\zeta_{\vecn},\mu_{\vecn})$ and $(\zeta_\nmax,\mu_\nmax)$, respectively, i.e., distributions satisfying both Eq.~\eqref{Eq:QuantumCor} and Eq.~\eqref{Eq:Axes2}. 

\subsection{Preliminary characterization of quantum correlations in the 2-dimensional space $(\zeta_{\vecn}, \mu_{\vecn})$}
\label{App:2DSpace}

Clearly, the maximizations involved in Eq.~\eqref{Eq:Goal} require some level of characterization of the 2-dimensional projection of the set of quantum distributions defined by $(\zeta_{\vecn}, \mu_{\vecn})$ for general $\vecn$. Let us denote this set by $\Qn$. A few remarks about $\Qn$ are now in order (see Figure~\ref{Fig:2dSlice}).

\begin{figure}[h]
  \includegraphics{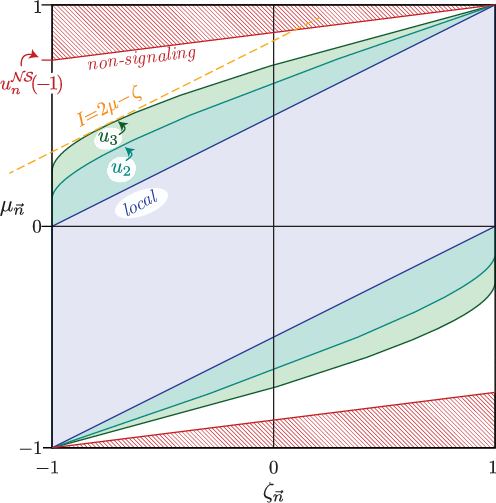}
  \caption{\label{Fig:2dSlice} A 2-dimensional projection of the sets of correlations onto the $(\CorLG,\CorAG)$-plane. The set of Bell-local (i.e., 1-producible) correlation is the parallelogram defined by the four extreme points $(-1,-1)$, $(-1,0)$, $(1,1)$ and $(1,0)$, whereas the set of $n$-partite correlations constrained only by the non-signaling conditions is the larger parallelogram defined by the four extreme points 
   $(-1,-1)$, $(-1, u_n^{\mbox{\tiny $\mathcal{NS}$}}(-1))$, $(1,1)$ and $(1, -u_n^{\mbox{\tiny $\mathcal{NS}$}}(-1))$.
  }
  
\end{figure}

\begin{enumerate}
	\item $\Qn$ is a convex set\footnote{This follows from Eq.~\eqref{Eq:QuantumCor}, Eq.~\eqref{Eq:Axes2} and the definition of a convex set. The proof is analogous to the proof of the convexity of $\Q_n$ (see, for instance, Ref.~\cite{Liang:Thesis}.)} and is invariant under reflection about the origin $(0,0)$.
	\item $\Qn$ is a superset of the set of local correlations and thus contains the points $(-1,0)$, $(1,0)$ as well as the {\em extreme} points $(1,1)$ and $(-1,-1)$.
	\item\label{Prop:Mono} For any given  $\CorLG$, the boundary of $\Qn$ is specified by
		\begin{equation}
		  -\ubGg{-\zeta_{\vecn}} \le  \CorAG \le  \ubG{{\vecn}} ,
		\end{equation}
		with $\ubG{{\vecn}}: [ -1,1 ] \rightarrow [ 0,1 ]$ being a concave, continuous function satisfying
		$ \ubGg{-1}\ge 0$ and $\ubGg{1}=1$. Since $\ubG{{\vecn}}\le 1$, its concavity also implies that $\ubG{{\vecn}}$ is 			
		{\em non-decreasing} with $\CorLG$.
	\item \label{Prop:Majorize} Any quantum distribution $(\zeta_n, \mu_n)$ achievable in the $n$-partite scenario is also achievable in an $n'$-partite scenario whenever $n'>n$.\footnote{To see this, it suffices to have $n$ of the parties sharing an $n$-partite entangled state, while the rest of the parties measure only the trivial observable $\id$.} Thus, for any given $y\in[-1,1]$, $\ub{{n'}}{y}\ge \ub{n}{y}$ and hence $\ub{\vecn}{y} \ge \ub{\nmax}{y}$.
\end{enumerate}
Since $\SqkM{\vecn}$ involves the maximization of a linear function of $\CorAG$ and $\CorLG$, the maximum must occur at one of the extreme points of $\Qn$. Thus, from Eq.~\eqref{Eq:vecnMax} and remark~\ref{Prop:Mono} above, we obtain
\begin{equation}\label{Eq:SqkMvecn}
	\SqkM{\vecn}
	=\max_{\CorLG}\quad 2\ubG{\vecn} -\CorLG,
\end{equation}
and analogously
\begin{equation}\label{Eq:SqkMvecn2}
	\SqkM{\nmax}
	=\max_{\CorLG}\quad 2\ub{\nmax}{\CorLG} -\CorLG.
\end{equation}
From remark~\ref{Prop:Majorize} and assumption~\eqref{Eq:nmax}, we also have
\begin{equation}\label{Eq:Majorization2}
	\SqkM{\vecn}\ge \SqkM{\nmax}.
\end{equation}

\subsection{A lemma relating $\Q_\vecn$ and $\Q_\nmax$}
\label{App:Lemma}
The proof of Theorem~\ref{Thm:Main}, or equivalently Eq.~\eqref{Eq:Goal} then follows from the following lemma and some further characterizations of $\Q_n$.
\begin{lemma}\label{Lemma:Max}
	For $\vecn=\{n_1,n_2,\ldots,n_m\}$ and $-1\le \CorLG\le 0$, $\ubg=\max_i \{ u_{n_i}(\zeta_{\vecn})\}$.
\end{lemma}

\begin{proof}

By definition, for any given $\CorLG$, we have
\begin{subequations}
\begin{align}
	\ubg=&\max_{\{\zeta_{n_i}\}} \prod_{i=1}^m \mu_{n_i}(\zeta_{n_i})=\max_{\{\zeta_{n_i}\}} \prod_{i=1}^m u_{n_i}(\zeta_{n_i}), \label{Eq:Obj}\\
	\text{s.t.}\quad \CorLG=& \prod_{i=1}^m \zeta_{n_i}.\label{Eq:zeta_constr}
\end{align}
\end{subequations}
In particular, for any given $\CorLG<0$, we see from Eq.~\eqref{Eq:zeta_constr} that there is a strict subset of $\{\zeta_{n_i}\}_{i=1}^m$ --- which we denote by $\mathcal{Z^-}$ --- such that $\zeta_{n_i}<0$ for all $\zeta_{n_i}\in\mathcal{Z^-}$. Clearly, $\left|\mathcal{Z^-}\right|$, the number of elements in $\mathcal{Z^-}$ must be an odd number. Recall from remark~\ref{Prop:Mono} above that each $u_{n_i}(\zeta_{n_i})$ is monotonically non-decreasing with $\zeta_{n_i}$, thus if $\left|\mathcal{Z^-}\right|\ge 3$, then for any given $\{\zeta_{n_i}\}_{i=1}^m$ that satisfies Eq.~\eqref{Eq:zeta_constr}, we may flip the sign of {\em any pair} of $\zeta_{n_i}\in\mathcal{Z^-}$, and the corresponding value of the objective function defined in Eq.~\eqref{Eq:Obj} is at least as large as before while the new values of $\{\zeta_{n_i}\}_{i=1}^m$ still satisfy Eq.~\eqref{Eq:zeta_constr}. Applying this iteratively, we eventually end up with $\mathcal{Z^-}=\{\zeta_{n_j}\}$ for some $j\in\{1,\ldots,m\}$, i.e., $\mathcal{Z^-}$ is now a singleton set.

Consider, for instance, the case where $\zeta_{n_1}<0$ while $\zeta_{n_i}>0$ for all $i\neq1$. Again, from the monotonicity of the function $u_{n_i}(\zeta_{n_i})$ and the fact that $|\zeta_{n_i}|\le1$ for all $i$, we see that we may increase the value of the product $\prod_{i=1}^m u_{n_i}(\zeta_{n_i})$ in Eq.~\eqref{Eq:Obj} by setting $\zeta_{n_1}\to\CorLG$ and $\zeta_{n_i}\to1$ for all $i\neq1$ while preserving the constraint given in Eq.~\eqref{Eq:zeta_constr}. Since $u_{n_i}(1)=1$ for all $n_i$, with these new values of $\zeta_{n_i}$, we see that the product becomes $\prod_{i=1}^m u_{n_i}(\zeta_{n_i})=u_{n_1}(\CorLG)$. Carrying out similar analysis for all possible singleton sets of $\mathcal{Z}^-$, we thus arrive at:
\begin{equation}\label{Eq:LHS}
	\ubg=\max_i \{u_{n_i}(\CorLG)\}=\ub{\nmax}{\CorLG}
\end{equation}
whenever $\CorLG<0$. The case where $\CorLG=0$ can be treated similarly and thus we arrive at Eq.~\eqref{Eq:LHS} whenever $-1\le \CorLG\le 0$.
\end{proof}

Note that in Eq.~\eqref{Eq:LHS}, the only quantity that matters in $\vecn$ is $\nmax$. Thus, if we can show that it suffices to consider $\CorLG\in[-1,0]$, we will have completed the proof of Theorem~\ref{Thm:Main} via Eq.~\eqref{Eq:Goal}, Eq.~\eqref{Eq:SqkMvecn} and Eq.~\eqref{Eq:SqkMvecn2}. We now  show that this is indeed the case by considering three distinct scenarios.

\subsection{Completing the proof via further characterization of $\Q_\nmax$ (and $\Q_\vecn$)}
\label{App:FurtherCharacterization}
\subsubsection{$\nmax\ge 6$ }

In this case, we see from Eq.~\eqref{Eq:Majorization2} that 
$\SqkM{\vecn}\ge\SqkM{6}>2$  (see 
Table~\ref{Tbl:QuantumMax}) whereas for $\CorLG>0$,
\begin{equation}
	\max_{\CorLG| \CorLG>0}\quad 2\ubG{\vecn} -\CorLG < \max_{\CorLG}\quad 2\ubG{\vecn} \le 2.
\end{equation}
Hence, in the maximization of Eq.~\eqref{Eq:SqkMvecn},  $\SqkM{\vecn}$ is attained only when $\CorLG\le0$. 
 
 \subsubsection{$3\le \nmax\le 5$ }
 
From the hierarchy of semidefinite programs~\cite{QMP:Hierarchy1,QMP:Hierarchy2}, we could certify that for $n=3,4,5$, the following choice of observables 
\begin{gather*}
	A_{x_i=1}=\cos\frac{\pi}{n}\,\sigma_x + \sin\frac{\pi}{n}\,\sigma_y,\nonumber\\
	A_{x_i=0}=\cos\left[\frac{3n+1}{n(n+1)}\pi\right]\,\sigma_x + \sin\left[\frac{3n+1}{n(n+1)}\pi\right]\,\sigma_y,
\end{gather*}
together with the quantum state $\ket{\mbox{GHZ$_n$}}$, give rise to the extreme point of $\Q_n$ at $\zeta_n=-1$. Likewise, the same quantum state in conjunction with the following choice of observables 
\begin{gather*}
	A_{x_i=1}=\cos\frac{\pi}{2n}\,\sigma_x + \sin\frac{\pi}{2n}\,\sigma_y,\nonumber\\
	A_{x_i=0}=\cos\left[\frac{1-n}{2n(n+1)}\pi\right]\,\sigma_x + \sin\left[\frac{1-n}{2n(n+1)}\pi\right]\,\sigma_y,
\end{gather*}
give rise to the extreme point of $\Q_n$ at $\zeta_n=0$. Explicitly, this means that we have
\begin{gather*}
	u_3(-1)= 0.2500,\quad u_4(-1)= 0.3466,\quad u_5(-1)= 0.4219,\\
	u_3(0)= 0.7286,\quad u_4(0)= 0.7781,\quad u_5(0)= 0.8122.
\end{gather*}

Hence, for $\nmax=3,4,5$, we have $u_\nmax(0)-u_\nmax(-1)<\frac{1}{2}$. By Lemma~\ref{Lemma:Max}, we see that for $\CorLG\in[-1,0]$, $\ubG{\vecn}=u_\nmax(\CorLG)$. Recall from remark~\ref{Prop:Mono} above that $\ubG{\vecn}$ is a concave function of $\CorLG$, the lemma below thus allows us to conclude that $\SqkM{\vecn}$ is attained when $\CorLG\le0$.
\begin{lemma}\label{lem:maxleft}
	Let $g(y): [-1,1]\to [0,1]$ be a concave function that satisfies $g(0)-g(-1) <\frac{1}{2}$, then the maximum value of $h(y)=2g(y)-y$ is attained at $y\in[-1,0]$.
\end{lemma}

\begin{proof}
Since $g(y)$ is concave in $y$, it follows from the defining property of a concave function that
\begin{align*}
	g(t\,y_3+(1-t)\,y_1)\ge t\,g(y_3)+(1-t)\,g(y_1)
\end{align*}
for arbitrary $t\in(0,1)$. In particular, for any $y_1<y_2<y_3$ with $y_i\in[-1,1]$, let us set $t=\frac{y_2-y_1}{y_3-y_1}\in(0,1)$, and it follows from the above inequality that
\begin{gather*}
	g(y_2)- t\,g(y_3)-(1-t)\,g(y_1) \ge0,\\
	\Rightarrow  g(y_2)(y_3-y_1) -g(y_3)(y_2-y_1)-g(y_1)(y_3-y_2) \ge 0,\\
	\Rightarrow  \left[g(y_2)-g(y_1)\right](y_3-y_2) \ge \left[g(y_3)-g(y_2)\right](y_2-y_1),\\
    \Rightarrow\frac{g ( y_{2} ) -g ( y_{1} )}{y_{2} -y_{1}} \ge
    \frac{g ( y_{3} ) -g ( y_{2} )}{y_{3} -y_{2}} .
\end{gather*}
  For $y_{1} =-1$, $y_{2} =0$ and $y_{3} =y>0$, we obtain from the above inequality that $g (
  y ) -g ( 0 ) \le y [ g ( 0 ) -g ( -1 ) ] <\frac{y}{2}$ where the last inequality follows from our initial assumption. Thus
  $h( y ) - h( 0 ) =2 [ g ( y ) -g ( 0 ) ] -y \le 0$. In other words, the value of the function $h(y)$ for any $y>0$ is smaller than equal to the value of the function at $h(0)$, which proves that the maximum of $h(y)$ must occur at $y\in[-1,0]$.

\end{proof}

\subsubsection{$\nmax=2$}

Let us first note that the boundary of $\Q_2$ can be completely characterized and reads as:
\begin{equation}
	\ub{2}{\zeta_{2}}=\cos^3\left(\frac{\arccos\zeta_{2}}{3}\right).
\end{equation}
\begin{proof}
Let us denote by $A_i$ and $B_j$ the dichotomic observable measured, respectively, by the first and second party. For any quantum state $\rho$ with expectation values,\footnote{To simplify the presentation, we  omit the tensor product symbol $\otimes$ in all subsequent equations.}
\begin{equation}
\begin{split}
	\mu_2&=\frac{1}{4}\tr\left[\rho\,\left(A_0B_0+A_1B_0+A_0B_1+A_1B_1\right)\right],\\
	 \zeta_2&=\tr(\rho\,A_1B_1),
\end{split}
\end{equation}
where $A_0^2=A_1^2=B_0^2=B_1^2=\id$,
it can be verified that
\begin{align}\label{Eq:SOS}
	\cos^3\left(\frac{\arccos\zeta_{2}}{3}\right)&-\mu_2\\ 
	=\frac{1}{16\lambda_+\lambda_-}\Big\{&\lambda_+\tr\left[\rho(\lambda_-A_0 - A_1 +\lambda_-B_0-B_1)^2\right]\nonumber\\
	+&\lambda_-\tr\left[\rho(\lambda_+A_0 + A_1 -\lambda_+B_0-B_1)^2\right]\Big\}\nonumber
\end{align}
where $\lambda_\pm=2\cos\left(\frac{\arccos\zeta_{2}}{3}\right)\pm1$. Since $\zeta_2\in[-1,1]$, $\lambda_\pm\ge 0$ and thus the right-hand-side of Eq.~\eqref{Eq:SOS} is always non-negative. In other words, $\cos^3\left(\frac{\arccos\zeta_{2}}{3}\right)$ is indeed a legitimate upper bound on the expectation value $\mu_2$. Moreover, this upper bound is achievable, for instance, 
 by performing the measurements:
\begin{gather*}
	A_0=B_0=\cos\left(\frac{\arccos\zeta_2}{6}\right)\sigma_x-\sin\left(\frac{\arccos\zeta_2}{6}\right)\sigma_y,\\
	A_1=B_1=\cos\left(\frac{\arccos\zeta_2}{2}\right)\sigma_x+\sin\left(\frac{\arccos\zeta_2}{2}\right)\sigma_y,
\end{gather*}
on the Bell state $\ket{\Phi^+}=\ket{\mbox{GHZ$_2$}}$.  
\end{proof}

Recall from Lemma~\ref{Lemma:Max} that for $\nmax=2$ and for $\CorLG\in[-1,0]$, we have $\ubg=\ub{2}{\CorLG}$. Since the objective function $2\ubg -\CorLG$ corresponds precisely to a hyperplane of gradient $\frac{1}{2}$ on this plane, and that $\ubg$ is a concave function of $\CorLG$, we know that if there is some value of $\CorLG=\CorLG^*$ such that the gradient of $\ubg$ is $\frac{1}{2}$, the maximum of $2\ubG{\vecn} -\CorLG$ is already attained at $\CorLG^*$.\footnote{This is true even though we do not have the explicit characterization of $\ubg$ for $\CorLG>0$.} Indeed, at $\CorLG=-\frac{1}{\sqrt{2}}$, one can verify that the gradient of $\ubg=\ub{2}{\CorLG}$ is exactly $\frac{1}{2}$, which means that $\SqkM{\vecn}$ is attained at $\CorLG<0$.

\section{The $k$-producible non-signaling bounds}
\label{App:NS}

For the set of non-signaling ($\mathcal{NS}$) correlations, it is easy to see that its boundary in the 2-dimensional plane depicted in Fig.~\ref{Fig:2dSlice} is given by:
\begin{equation}
	u_n^{\mbox{\tiny $\mathcal{NS}$}}(\zeta_n)=1+2^{-n}(\zeta_n-1).
\end{equation}
For extremal non-signaling distribution, it is easy to se that Eq.~\eqref{Eq:Axes2} holds.
By applying Lemma~\ref{Lemma:Max} to the non-signaling correlations, we see that for $\CorLG\in[-1,0]$, we must also have
\begin{equation}
	u_{\vecn}^{\mbox{\tiny $\mathcal{NS}$}}(\CorLG)=1+2^{-\nmax}(\CorLG-1).
\end{equation}
Since this segment of the boundary is a straight line with gradient $2^{-\nmax}<\frac{1}{2}$ for all $\nmax>1$, and by the concavity of $u_{\vecn}^{\mbox{\tiny $\mathcal{NS}$}}$, we know that 
\begin{equation}\label{Eq:NSBound}
\begin{split}
	\max_{\mbox{\tiny k-pr. $\mathcal{NS}$}} \Sn 
	= & \max_{\CorLG}\quad 2u_{\vecn}^{\mbox{\tiny $\mathcal{NS}$}}(\CorLG)-\CorLG,\\ 
	= & 2u_{\vecn}^{\mbox{\tiny $\mathcal{NS}$}}(-1)+1=3-2^{2-\nmax},
\end{split}
\end{equation}
where $k=\nmax$. Since $\I_n$ is an inequality involving only full correlators, it is worth noting~\cite{Curchod:IP} that the bound derived in Eq.~\eqref{Eq:NSBound} also holds for more powerful (e.g. signaling) resource.

\section{Maximal possible quantum violation of the MABK inequality by $k$-producible states}
\label{App:MABK}

Here, we give a proof that the maximal possible quantum violation of the MABK inequality by 2-producible quantum states is $\sqrt{2}$ regardless of $n$. To see this, it suffices to consider a partition of the $n$ parties into groups of 2 parties whenever possible. If $n$ is even, we then have $m=\frac{n}{2}$ and $L=0$, the result of Nagata~{\em et al.}~\cite{Nagata:PRL:2002} then implies that the maximal possible quantum violation in this case is $2^{(n+L-2m+1)/2}=\sqrt{2}$. Similarly, if $n$ is odd, we have $m=\frac{n+1}{2}$ and $L=1$, giving also the same bound of $\sqrt{2}$. Consequently, the witness given by inequality~\eqref{Ineq:MABKWitness} is also legitimate for $k=2$ and arbitrary $n\ge k$. The bound for $(n-1)$-producible states, namely, $2^{(n-2)/2}$ also follows straightforwardly from similar analysis (see~\cite{ANL} for an alternative proof).

For larger value of $k$ with $\lceil \frac{n}{2} \rceil \le k<n-1$, it is easy to see that the maximal possible quantum violation of the MABK inequality is achieved by having the $n$ parties separated into a group of size $k$ and a group of size $n-k$. In this case, $m=2$, $L=0$ (for $n>2$) and we have the maximal possible MABK-inequality violation of $2^{(n-3)/2}$ for all $k$ with $\lceil \frac{n}{2} \rceil \le k<n-1$. As a result, the inequality~\eqref{Ineq:MABKWitness} generally does not hold for $n\ge6$. The table below gives a clear illustration of this fact for $n=6$ ($k=3,4$) as well as $n=7$ ($k=4,5$).

\begin{table}[h!]
\begin{ruledtabular}
	\begin{tabular}{c|c|c|c|c|c}
	 $n$  & $m$  & $L$ & Partition size & $2^{(n+L-2m+1)/2}$ & ED \\ \hline
	 3      & 1 	 & 0     & $\{3\}$   & $2$                  & 3 \\ \hline
	 4      & 1 	 & 0     & $\{4\}$   & $2\sqrt{2}$                  & 4 \\ 
	 4      & 2 	 & 1     & $\{3,1\}$   & $2$                  		& 3 \\ \hline
	 5      & 1 	 & 0     & $\{5\}$   & $4$                  	& 5 \\ 
	 5      & 2 	 & 1     & $\{4,1\}$   & $2\sqrt{2}$                 		& 4 \\ 
	 5      & 2 	 & 0     & $\{3,2\}$   & $2$                  & 3 \\ \hline
	 6      & 1 	 & 0     & $\{6\}$   & $4\sqrt{2}$                  	& 6 \\ 
	 6      & 2 	 & 1     & $\{5,1\}$   & $4$                 		& 5 \\ 
	 6      & 2 	 & 0     & $\{4,2\}$   & $2\sqrt{2}$                  & 4 \\ 
	 6      & 2 	 & 0     & $\{3,3\}$   & $2\sqrt{2}$                  & 3 \\ \hline
	 7      & 1 	 & 0     & $\{7\}$   & $8$                  	& 7 \\ 
	 7      & 2 	 & 1     & $\{6,1\}$   & $4\sqrt{2}$                 		& 6 \\ 
	 7      & 2 	 & 0     & $\{5,2\}$   & $4$                  		& 5 \\ 
	 7      & 2 	 & 0     & $\{4,3\}$   & $4$                  		& 4 \\ 
	 7      & 3 	 & 1     & $\{3,3,1\}$   & $2\sqrt{2}$                  & 3 \\  
	\end{tabular}
	\caption{Maximal possible MABK-inequality violation   by $n$-partite quantum states admitting different equivalences classes of partitions --- parameterized by $m$, the number of groups, and $L$, the number of unentangled subsystems. Different classes are specified by different combinations of $m$ integers (separated by commas) in a curly bracket, each representing the number of constituent subsystems in a group. For instance, the partition size $\{3,2\}$ represents a 5-partite quantum state that is formed by the tensor product of a tripartite quantum state and a bipartite quantum state. Here, for each given value of an entanglement depth (ED), only one class of partition achieving the maximum MABK-inequality violation is shown; also we omit partitions corresponding to  ${\rm ED}\le2$.}
\end{ruledtabular}		
\end{table}

\section{Explicit values of quantum violation}
\label{App:QuantumViolation}

\begin{table}[h!]
\begin{ruledtabular}
	\begin{tabular}{c|cccccc}
		$n$  &  $2$ & $3$ 
		& $4$ & $5$ & $6$ & 7 \\ \hline
		 $\ket{\mbox{GHZ$_n$}}$ &   $\sqrt{2}$  & $\frac{5}{3}$ & 1.8428  &  1.9746  &   2.0777  & 2.1610
\vspace{0.1cm}\\ 
		 $\ket{\mbox{W$_n$}}$ &  $\sqrt{2}$  & 1.3631 & 1.3633  &  1.3656  &   1.3677   & 1.3693 \vspace{0.1cm}\\
		 $\ket{C^-_n}$ &   $\sqrt{2}$  & $\frac{5}{3}$ & $\sqrt{2}$  &  $\sqrt{2}$  &  $\sqrt{2}$ &  $\sqrt{2}$   \vspace{0.1cm}\\
		 $\ket{C^o_n}$ &   $\sqrt{2}$  & $\frac{5}{3}$ & $\sqrt{2}$  &  $1.1535$  &  $1.1583$   & 1.1563 \vspace{0.1cm}\\ \hline
		 $\ket{\mbox{GHZ$_n$}}$ &   $\sqrt{2}$  & $2$ & $2\sqrt{2}$  &  4 & $4\sqrt{2}$ & 8  \vspace{0.1cm}\\
		 $\ket{\mbox{W$_n$}}$ &     $\sqrt{2}$  & 1.5230 & 1.5543  & 1.5698  &  1.5794 & 1.5859  \vspace{0.1cm}\\
		 $\ket{C^-_n}$ &     $\sqrt{2}$  & 2 & $\sqrt{2}$  &  $\sqrt{2}$  &  $\sqrt{2}$   &  $\sqrt{2}$   \vspace{0.1cm}\\
		 $\ket{C^o_n}$ &     $\sqrt{2}$  & {2} & $\sqrt{2}$  &  $1$  &  1 &  1  \vspace{0.1cm}\\		 
		\end{tabular}
		\caption{Best quantum violation of $\I_n$ (top block) and the MABK inequality (bottom block) found for the GHZ state, the W-state, and the 1-dimensional cluster states $\ket{C^-_n}$ and $\ket{C^o_n}$. These quantum violations were obtained with the help of the algorithm described in~\cite{Liang:PRA:2007}.}
\end{ruledtabular}		
		\end{table}

\section{Useful semidefinite programs}
\label{App:SDP}

Here, we provide details of the semidefinite program (SDP) alluded to in the main text.
In~\cite{Moroder:PRL:PPT}, it was shown that a certain matrix of expectation values $\chi$ containing experimentally accessible quantities $\{P(\vec{a}|\vec{x})\}$ can be seen as the result of a local completely-positive map $\Lambda$ acting on the underlying quantum state $\rho$. Since {\em no} local mapping can increase entanglement depth (i.e., it cannot make a $k$-producible $\rho$ not $k$-producible), if $\chi$ is not $k$-producible, so is the underlying state $\rho$. In other words, certifying that $\chi$ is not $k$-producible also certifies that the quantum state that gives rise to the correlation $\{P(\vec{a}|\vec{x})\}$ must have an entanglement depth of at least $k+1$. To this end, let us remind that if a quantum state is separable with respect to a certain partitioning of the system into subsystems, the corresponding partial transposition(s)~\cite{Peres:PRL:1996} of the quantum state must remain positive semidefinite. These observations together then allow us to (1) upper-bound the maximal possible quantum violation of {\em any} given (linear) Bell inequality by $k$-producible quantum states, and (2) certify  (through the relaxation of partial transposition mentioned above) that some given correlation $\{P(\vec{a}|\vec{x})\}$ does not originate from $k$-producible quantum states.

\subsection{SDP for upper bounding quantum violation by $k$-producible states}

Recall from~\cite{Moroder:PRL:PPT} that at any given level (say, $\ell$) of the hierarchy considered therein, we consider a matrix $\chi_\ell[\rho]$ that can be decomposed as:
\begin{equation}\label{Eq:Decomposition}
	\chi_\ell[\rho]=\sum_{\vec{a},\vec{x}} P(\vec{a}|\vec{x})\,F^\ell_{\vec{a},\vec{x}} + \sum_v u_v F^\ell_v,
\end{equation}
i.e., into one fixed part that linearly depends on the experimentally accessible quantities $\{P(\vec{a}|\vec{x})\}$, and a complementary (orthogonal) part that is known only if the underlying state $\rho$ and the measurement giving rise to the correlation $\{P(\vec{a}|\vec{x})\}$ is known; in Eq.~\eqref{Eq:Decomposition}, $F^\ell_{\vec{a},\vec{x}}$ and $F^\ell_v$ are some fixed, symmetric, Boolean matrices~\cite{Moroder:PRL:PPT}.

For any {\em given} partition $\mathcal{P}$ of the $n$ parties into subsets  of $m$ groups satisfying Eq.~\eqref{Eq:nmax} and for any given Bell inequality specified by $\sum_{\vec{a},\vec{x}}\beta^{\vec{x}}_{\vec{a}}P(\vec{a}|\vec{x})$, the $\ell$-level upper bound that we desire can be obtained by solving the following  optimization problem:
\begin{align}
	\max_{\{P(\vec{a}|\vec{x})\},\{u_v\}}\quad 
	&\sum_{\vec{a},\vec{x}}\beta^{\vec{x}}_{\vec{a}}P(\vec{a}|\vec{x}),\label{Eq:SDPMaxViolation}\\
	{\rm s.t.} \quad\qquad &\chi^{\phantom{T_j}}_\ell\!\!\![\rho]\ge 0, \quad \chi_\ell^{\mbox{\tiny T$_j$}}\ge0\quad\forall\, j\in{\rm Com}(\mathcal{P})\nonumber
\end{align}
where ${\rm Com}(\mathcal{P})$ refers to the set of all groupings of parties that are compatible with $n$ and $k$. For example, when $n=6$ and $k=3$, a possible partition of 7 parties is given by $\mathcal{P}=\{\{1,2,3\},\{4,5,6\},\{7\}\}$, i.e., the first 3 parties are entangled, and the same applies for the next 3. In this case, the constraints in Eq.~\eqref{Eq:SDPMaxViolation} would include the requirement that the partial transpositions $\chi_\ell^{\mbox{\tiny T$_{123}$}}$, $\chi_\ell^{\mbox{\tiny T$_{456}$}}$, $\chi_\ell^{\mbox{\tiny T$_{7}$}}$ are positive semidefinite. Since the optimization problem given in Eq.~\eqref{Eq:SDPMaxViolation} only involves matrix positivity constraints and the objective function is linear in the matrix entries, it is thus an SDP.  The desired upper bound is then obtained by taking the maximum of all such upper bounds (returned by the SDPs) when considering all different partitions consistent with the assumption of Eq.~\eqref{Eq:nmax}. 

\subsection{SDP for determining if a given correlation can be produced by $k$-producible states}

For any given correlation $P_{\mbox{\tiny obs}}(\vec{a}|\vec{x})$, determining its compatibility with $k$-producible states can be achieved, instead, by solving the following SDP:
\begin{equation}\label{Eq:SDPk-prod}
\begin{split}
	{\rm Find}\quad  &\{u_v\},\{P^{(j)}(\vec{a}|\vec{x})\},\{u_v^{(j)}\}\\
	{\rm s.t.} \quad 
	&\chi_\ell=\sum_{\vec{a},\vec{x}} P_{\mbox{\tiny obs}}(\vec{a}|\vec{x})\,F^\ell_{\vec{a},\vec{x}} + \sum_v u_v F^\ell_v\ge 0, \\
	&\chi_\ell=\sum_{j} \chi_\ell^{\P_j},\\
	&\chi_\ell^{\P_j}=\sum_{\vec{a},\vec{x}} P^{(j)}(\vec{a}|\vec{x})\,F^\ell_{\vec{a},\vec{x}} + \sum_v u_v^{(j)} F^\ell_v\ge0, \\
	&\left(\chi_\ell^{\P_j}\right)^{\mbox{\tiny T$_i$}}\ge0 \quad  \forall\, i\in{\rm Com}(\P_j),
\end{split}
\end{equation}
where $\{\P_j\}$ is the set of all possible partitions of $n$ parties that are consistent with Eq.~\eqref{Eq:nmax}.

\section{Generalization of $\I_n^k$ to $\I_n^k(\gamma)$}
\label{App:OtherBellIneq}

Consider the Bell expression:
\begin{equation}\label{Eq:SnGamma}
	\Sn^\gamma=\frac{\gamma}{2^n}\sum_{\vec{x}\in\{0,1\}^n} E_n(\vec{x})-E_n(\vec{1}_n). 
\end{equation}
Denote by $\S^{\Q,*}_{n,\gamma}$ and $\S^{\mbox{\tiny $\Q$},*}_{\mbox{\tiny $k$-pr.}, \gamma}$ the maximal possible quantum value of $\Sn^\gamma$ achievable by, respectively,  $n$-partite quantum states in general and those which are $k$-producible. Here, we give a proof that for $0<\gamma\le 2$,  $\S^{\mbox{\tiny $\Q$},*}_{\mbox{\tiny $k$-pr.}, \gamma}=\S^{\Q,*}_{k,\gamma}$, which generalizes Theorem~\ref{Thm:Main} (corresponding to to the case of $\gamma=2$) to a one-parameter family of Bell expressions. Essentially, this follows from the fact that in Fig.~\ref{Fig:2dSlice}, $\Sn^\gamma$ with $\gamma<2$ corresponds to a steeper hyperplane compared to what we have already established for $\gamma=2$. And intuitively, with a steeper objective function, the maximizer cannot move to the right, thus allowing us to conclude with the help of Lemma~\ref{Lemma:Max}.

To this end, we shall first prove a lemma that allows us to relate the maximizer(s) of $\S^{\mbox{\tiny $\Q$},*}_{\mbox{\tiny $k$-pr.}, \gamma}$ for different values of $\gamma$. 
Let $g(y):\mathbb{R}\rightarrow\mathbb{R}$ be a bounded, real-valued function defined on some subset of $\mathbb{R}$. Consider now the function $h_s(y)=g(y)-s\,y$ and let $y_s$ be a maximizer of $h_s(y)$, i.e., $h_s(y_s)$ is a global maximum of $h_s(y)$, then the following lemma relates the maximizers for different values of $s$.

\begin{lemma}\label{Lem:ToTheLeft}
For $s'>s>0$ and $y>y_s$, it holds that $h_{s'}(y_s)-h_{s'}(y)>0$.
\end{lemma}
\begin{proof}
Since $y_s$ gives rise to a global maximum of $h_s$, we have that for all $y$ in the domain of $g(y)$,
\begin{align}
& h_s(y_s)-h_s(y)\geq 0\nonumber\\
\Leftrightarrow\,\, &g(y_s)-g(y)+s(y-y_s)\geq 0.
\label{Eq:GlobalMax}
\end{align}
Then for $y>y_s$ and $s'>s>0$
\begin{equation}
\begin{split}
h_{s'}(y_s)-h_{s'}(y)&=g(y_s)-g(y)+s'(y-y_s)\\
&>g(y_s)-g(y)+s(y-y_s)\\
&\geq 0,
\end{split}
\end{equation}
where the first inequality follows from the assumptions that $s'>s$ and $y>y_s$ and the second follows Eq.~\eqref{Eq:GlobalMax}.
\end{proof}

The above lemma then allows us to show that the DIWED that we provided in Eq.~\eqref{Eq:EDWitnesses} is actually a special case of the more general family of witnesses:
\begin{equation}\label{Eq:EDWitnessesGeneralized}
	\I_n^k(\gamma): \frac{\gamma}{2^n}\sum_{\vec{x}\in\{0,1\}^n} E_n(\vec{x})-E_n(\vec{1}_n)\,  \stackrel{\stackrel{\mbox{\tiny $k$-producible}}{\mbox{\tiny states}}}{\le}  \S^{\Q,*}_{k,\gamma},
\end{equation}
where $0<\gamma\le 2$. In the notations of Appendix~\ref{App:ProofUB}, proving inequality~\eqref{Eq:EDWitnessesGeneralized} amounts to proving:
\begin{equation}\label{Eq:Goal2}
	\max_\vecn \max_{\CorLG}\quad \ubG{\vecn} -\gamma^{-1}\CorLG\quad 
	= \max_{\zeta_k}\quad u_k(\zeta_k)-\gamma^{-1}\zeta_k,
\end{equation}
under the assumption of Eq.~\eqref{Eq:nmax}. Recall from the proof of Theorem~\ref{Thm:Main} that for any given $\vecn$, the maximizer(s) of $\ubG{\vecn} -\frac{1}{2}\CorLG$ occur at $\CorLG\in[-1,0]$. Now, consider  $\gamma^{-1}> \tfrac{1}{2}$, Lemma~\ref{Lem:ToTheLeft} implies that the global maximizer of $\ubG{\vecn} -\frac{1}{\gamma}\CorLG$ also occur at $\CorLG\in[-1,0]$. An immediate application of Lemma~\ref{Lemma:Max} then allows us to conclude Eq.~\eqref{Eq:Goal2}, which in turn implies $\S^{\mbox{\tiny $\Q$},*}_{\mbox{\tiny $k$-pr.}, \gamma}=\S^{\Q,*}_{k,\gamma}$ and hence inequality~\eqref{Eq:EDWitnessesGeneralized}.

\end{document}